\documentclass[11pt,letterpaper,twoside]{article}
\usepackage[centertags]{amsmath}
\usepackage{graphicx,indentfirst,amsmath,amsfonts,amssymb,amsthm,newlfont}
\font\Goth=yinitas scaled \magstep0

\newcommand{\Gth}[1]{\lower2mm\hbox{\Goth #1}}

\def\al{\alpha}

\def\eps{\varepsilon}

\def\de{\delta}

\def\be{\beta}

\def\l1{{\lambda}_1}

\newcommand{\f}{\frac}

\def\x1{{\xi }_{xx}}
\def\x2{{\xi }_{yy}}
\def\x3{{\xi }_{xy}}

\def\e1{{\eta }_{xx}}
\def\e2{{\eta }_{yy}}
\def\e3{{\eta }_{xy}}

\newcommand{\ds}{\displaystyle }
\newtheorem{definition}{Definition}
\newtheorem{remark}{Remark}

\newtheorem{theorem}{Theorem}
\newtheorem{lemma}{\bf Lemma}
\newtheorem{example}{Example}

\newcommand{\beqn}{\begin{eqnarray*}}
\newcommand{\eeqn}{\end{eqnarray*}}
\newcommand{\beqnn}{\begin{eqnarray}}
\newcommand{\eeqnn}{\end{eqnarray}}
\newcommand{\p}{\partial}
\newcommand{\bb}{\begin{equation}}
\newcommand{\ee}{\end{equation}}
\newcommand{\ba}{\begin{array}}
\newcommand{\ea}{\end{array}}
\newcommand{\R}{\mathbb{R}}
\newcommand{\N}{\mathbb{N}}

\begin{document}
\pagenumbering{arabic}
\title{\huge \bf Symmetry analysis of a class of autonomous even-order ordinary differential equations}
\author{\rm \large Priscila Leal da Silva and Igor Leite Freire \\
\\
\it Centro de Matem\'atica, Computa\c c\~ao e Cogni\c c\~ao\\ \it Universidade Federal do ABC - UFABC\\ \it 
Rua Santa Ad\'elia, $166$, Bairro Bangu,
$09.210-170$\\\it Santo Andr\'e, SP - Brasil\\
\rm E-mail: priscila.silva@ufabc.edu.br/pri.leal.silva@gmail.com\\
\rm E-mail: igor.freire@ufabc.edu.br/igor.leite.freire@gmail.com}
\date{\ }
\maketitle
\vspace{1cm}
\begin{abstract}
{A class of autonomous, even-order ordinary differential equations is discussed from the point of view of Lie symmetries. It is shown that for a certain power nonlinearity, the Noether symmetry group coincides with the Lie point symmetry group. First integrals are established and exact solutions are found. Furthermore, this paper complements, for the one-dimensional case, some results in the literature of Lie group analysis of poliharmonic equations and Noether symmetries obtained in the last twenty years. In particular, it is shown that the exceptional negative power discovered in [{\sc Bokhari, A. H., Mahomed, F. M. and Zaman, F. D.} (2010) Symmetries and integrability of a fourth-order Euler-Bernoulli beam equation. {\em J. Math. Phys.}, \textbf{51}, 053517] is a member of a one-parameter family of exceptional powers in which the Lie symmetry group coincides with the Noether symmetry group.}
{Lie point symmetries, Noether symmetries, first integrals, exact solutions, ordinary differential equations, power nonlinearities.}
\\
2000 Math Subject Classification: 34A05, 35A30, 58J70, 70G65, 76M60
\end{abstract}
\vskip 1cm
\begin{center}
{2010 AMS Mathematics Classification numbers:\vspace{0.2cm}\\
34A05, 35A30, 58J70, 70G65, 76M60  \vspace{0.2cm} }
\end{center}
{\bf Dedicatory:} I. L. Freire dedicates this work for his father, Antonio Fernando Santos Freire.
\pagenumbering{arabic}
\newpage

\section{Introduction}

In this paper we consider the equation
\bb\label{1.1}
y^{(2n)}+f(y)=0
\ee
from the point of view of Lie group analysis.  In (\ref{1.1}) and from now on, $n$ is a positive integer, $x\in\R$ is an independent variable while $y=y(x)$ is a dependent one and $f$ is a smooth function. Moreover,
$$y':=\f{dy}{dx},\,y'':=\f{d^2y}{dx^2},\cdots,y^{(k)}:=\f{d^k y}{dx^k},\cdots$$

Such class of equations includes many important mathematical models for phenomena arising from Mathematical Physics and Engineering. For instance, when $n=1$ and $f(y)=\omega^2y$, equation (\ref{1.1}) is the well known harmonic oscillator. Another important equation is given by the celebrated Ermakov equation
\bb\label{1.2}
y''+\lambda y^{-3}=0,
\ee
which can physically be interpreted as an oscillator with a nonlinear restoring force acting on it. For $n=2$, equation (\ref{1.1}) models the applied load and the deflection in a beam when the acting force depends on the deflection, see \cite{han,bok1}. Special cases of equation (\ref{1.1}) are also employed for modeling phenomena in the general relativity and other Physics' branches, see \cite{go}, \cite{moyo}, \cite{moises}, \cite{pri} and references therein.

For a long time, the investigation of invariance properties of second order ODEs with power nonlinearities was very intense. In particular, equations of the type
$$y''=f(x)y^p,$$
were widely investigated in \cite{go}, \cite{so}, and, moreover, such an equation is linked to the general family
$$y''+p(x)y'+q(x)y=r(x)y^p$$
via Kummer-Liouville transformations, see \cite{go, mellin} and references therein for further details.

It is well known that a second order ordinary differential equation does not admit a $r-$dimensional symmetry Lie algebra if $r\in\{4,5,6,7\}$ and, moreover, its Lie algebra of symmetries is at most 8(=2+6) dimensional, see, for instance \cite{lie, fa1989}. This last case is reached when, physically, we have the free particle equation or then, the original equation is linearizable via point transformation, see \cite{fa1,sar}. Although these last sentences could suggest that only linear equations can admit an eight-dimensional Lie algebra of symmetries, it is well known that some nonlinear differential equations also have the same property, see \cite{sar} for a better discussion.

An arbitrary $n-th$ order linear ordinary differential equation possesses $n+1,\,n+2,\,n+3$ and the maximum $n+4$ symmetries when $n\geq3$ (see, for instance, \cite{fa1} and references therein), which shows a substantial difference compared with the one of second order. We also guide the interested reader to \cite{fa2}'s survey on Lie symmetry analysis of linear ordinary differential equations for further discussion and extensive bibliography regarding these points. 

An interesting case corresponds to equation (\ref{1.2}). Such an equation admits a three-dimensional, unsolvable, Lie algebra of symmetries isomorphic to $\frak{sl}(2,\R)$, see \cite{go}. Although it is not solvable, using such an algebra is sufficient to reduce (\ref{1.2}) to quadratures. Moreover, it is well known that the Lie symmetries of (\ref{1.2}) are also Noether symmetries and, therefore, it is possible to construct first integrals associated to each symmetry, see, for instance, \cite{moises}.

In a more recent paper, \cite{bok1} considered a fourth-order equation
$$y''''=f(y),$$
which is a nonlinear generalization of the static Euler-Bernoulli equation used to describe the relationship between the applied load and the deflection in a beam, see \cite{han,bok1}.

In the mentioned reference the authors carried out a complete group classification of that equation, as well as they considered its Noether symmetries and from then, first integrals.

Similarly to the second order differential equation, those authors showed that the equation
\bb\label{1.3}
y''''+\lambda y^{-\f{5}{3}}=0
\ee
also admits a three-dimensional symmetry Lie algebra, again isomorphic to the classical $\frak{sl}(2,\R)$ Lie algebra and all Lie symmetries are also Noether symmetries, which was a curious and surprising result, later discussed in \cite{bok2}. 

In a previous paper (\cite{ipm}, see also \cite{ipm2012}) we, jointly with M. Torrisi, considered the fourth-order equation
$$y''''+ax^\gamma y^p=0.$$
Among other results, we showed that in the nonlinear cases, the maximal symmetry Lie algebra is achieved when $\gamma=0$ and $p=-5/3$. This intriguing result motivated us to write the present paper.

In a na\"ive way, we observe the following:
$$-3=\f{1+2\times1}{1-2\times1},\,\,\,-\f{5}{3}=\f{1+2\times2}{1-2\times2}.$$
This simple observation shows a connection, at least to these two cases, between the order of the equation $2=2\times 1,\,4=2\times 2$ and the exceptional power. We will show, using Noether symmetries, that for the power
$$p=\f{1+2n}{1-2n}$$
all Lie point symmetries of the equation
$$y^{(2n)}+\lambda y^{p}=0,\,\,\lambda\neq0,$$
are Noether symmetries and, therefore, the mentioned results for (\ref{1.2}) and (\ref{1.3}) are consequences of our results.

The paper is as the follows. In the next section we present the main results (theorems \ref{teo1}, \ref{teo2} and \ref{teo3}), a preliminary discussion about the purposes of this work and the state of the art regarding these topics. A revision on some basic facts regarding Lie symmetries and Noether Theorem is done in section \ref{rev}. Section \ref{aux} presents proofs of some technical results, in order to avoid long and tedious demonstrations of theorems \ref{teo1}, \ref{teo2} and \ref{teo3}. Then, in section \ref{main}, the complete group classification of equation (\ref{1.1}) is carried out (Theorem \ref{teo1}). Theorems \ref{teo2} and \ref{teo3} are proved in section \ref{noether}. As a consequence, first integrals are established and some exact and explicit solutions are obtained in section \ref{integrals}. Final comments are presented  next.

\section{Main results and preliminary discussions}


In what follows, $\lambda$ is a real constant. Our first result is given by the following:
\begin{theorem}\label{teo1}
A basis to the Lie point symmetry generators of equation $(\ref{1.1})$, for an arbitrary $f=f(y)$ and any $n\geq1$, is generated by the vector field
\bb\label{1.4}
X_{1}=\f{\p}{\p x}.
\ee
For special choices of the function $f(y)$ it is possible to enlarge the Lie point symmetry group. The additional generators to $(\ref{1.4})$ are:
\begin{enumerate}
\item If $f(y)=\lambda e^{\al y},\,\lambda\al\neq0$ and $n\geq1$, we have
\bb\label{1.5}
D_{1}=x\f{\p}{\p x}-\f{2n}{\al}\f{\p}{\p y}.
\ee
\item If $f(y)=\lambda y^p,\,p\neq 0,1,\,\lambda\neq0$ and $n\geq1$, we have
\bb\label{1.6}
D_{p}=x\f{\p}{\p x}+\f{2n}{1-p}y\f{\p}{\p y}.
\ee
\item If $f(y)=\lambda y^{\f{1+2n}{1-2n}},\,\lambda\neq 0$ and $n\geq1$, we have
\bb\label{1.7}
X_{2}=x\f{\p}{\p x}+\f{2n-1}{2}y\f{\p}{\p y}
\ee
and
\bb\label{1.8}
X_{3}=x^2\f{\p}{\p x}+(2n-1)xy\f{\p}{\p y}.
\ee
\item If $f(y)=\lambda$:
\begin{enumerate}
\item if $n>1$, we have
\bb\label{l1}
Y_{1}=\ds{x\f{\p}{\p x}+\f{2n-1}{2}\left[y-\lambda\f{2n+1}{2n-1}\f{x^{2n}}{(2n)!}\right]}\f{\p}{\p y},
\ee
\bb\label{l2}
Y_{2}=\ds{x^2\f{\p}{\p x}+x\left[(2n-1)y-\lambda\f{x^{2n}}{(2n)!}\right]\f{\p}{\p y}},
\ee
\bb\label{l3}
Y_{3}=\ds{\left(y- \lambda \f{x^{2n}}{(2n)!}\right)\f{\p}{\p y}}
\ee
and
\bb\label{l4}
Z_{j}=\f{x^{j}}{j!}\f{\p}{\p y},\,\,\, 0\leq j \leq 2n-1.
\ee
\item if $n=1$, the additional generators are $(\ref{l1})$, $(\ref{l2})$, $(\ref{l3})$, $(\ref{l4})$ with $j=0,1$, 
and
\bb\label{l5}
Y_{4}=\left(xy-\f{\lambda}{2}x^3\right)\f{\p}{\p x}+\left(y^2-\f{\lambda}{4}x^4\right)\f{\p}{\p y},\,\,\,
Y_{5}=\left(y-\f{3}{2}\lambda x^2\right)\f{\p}{\p x}-\lambda x^3\f{\p}{\p y}.
\ee

\end{enumerate}

\item If $f(y)=\lambda y$,
\begin{enumerate}
\item if $n>1$, we have
\bb\label{1.10}V_{1}=y\f{\p}{\p y}\ee
and
\bb\label{beta0}
V_{\be}=\be(x)\f{\p}{\p y},
\ee
where $\be$ is a solution of the linear equation
\bb\label{beta}
\be^{(2n)}+\lambda\be=0.
\ee
\item if $n=1$, we have $(\ref{1.10})$, $(\ref{beta0})$ and
\bb\label{eqn=1}
\ba{l}
\ds{V_{2}=\sin{(2\sqrt{\lambda}x)}\f{\p}{\p x}+\sqrt{\lambda}y\cos{(2\sqrt{\lambda}x)}\f{\p}{\p y}},\,\,\,
\ds{V_{3}=\cos{(2\sqrt{\lambda}x)}\f{\p}{\p x}-\sqrt{\lambda}y\sin{(2\sqrt{\lambda}x)}\f{\p}{\p y}},\\
\\
\ds{V_{4}=y\sin{(\sqrt{\lambda}x)}\f{\p}{\p x}+\sqrt{\lambda}y^2\cos{(\sqrt{\lambda}x)}\f{\p}{\p y}},\,\,\,
\ds{V_{5}=y\cos{(\sqrt{\lambda}x)}\f{\p}{\p x}-\sqrt{\lambda}y^2\sin{(\sqrt{\lambda}x)}\f{\p}{\p y}.}
\ea
\ee
\end{enumerate}
\end{enumerate}
\end{theorem}

\begin{remark}
We observe that only for the linear cases there is a sensible difference between the values $ n=1$ or $n>1$. 
\end{remark}

\begin{remark}
For $p=(1+2n)/(1-2n)$ the operator (\ref{1.6}) coincides with the operator (\ref{1.7}).
\end{remark}

\begin{remark}
Although in Theorem \ref{teo1} we presented the symmetries of the linear cases, in the remaining of the paper we shall not consider them since we are mainly interested in nonlinear phenomena. However, we invite the interested reader to consult \cite{fa1,fa2} for a wider discussion on linear ordinary differential equations and Lie symmetries.
\end{remark}

\begin{remark}
Equation (\ref{beta0}) corresponds to a family of $2n$ Lie point symmetry generators parametrized by the $2n$ linearly independent solutions of the equation (\ref{beta}).
\end{remark}

\begin{remark}
The group classification of the case $n=1$ could be done by using computational packages, see, for instance, \cite{stelios1, stelios2}, and this can also be done for a given $n$. Even though case $n=1$ was already widely considered in the literature, we present it in the paper for sake of completeness.
\end{remark}

\begin{remark}\label{rem6}
The cases considered in Theorem \ref{teo1} are, in fact, the complete group classification modulo the equivalence transformations $\bar{x}=a_{1}x+a_{2}$, $\bar{y}=b_{1}y+b_{2}$, where $a_{1},\,a_{2},\,b_{1}$ and $b_2$ are constants such that $a_{1}b_{1}\neq0$, and $\bar{f}=(b_{1}/a_{1}^{2n})f$. We, however, will not proceed with further information on these transformations in this paper once we are interested in understanding the one-dimensional version of the results obtained in \cite{sv,yu1}. Moreover, it is well known that for semilinear equations such as (\ref{1.1}), in order to carry out a complete group classification, up to equivalence transformations, one must take the functions $f$ as considered in the Theorem \ref{1.1}, see \cite{bok1,yu1,yi1,yi4,ipm,sv}. They, in fact, arise from a compatibility condition that will be deduced on section \ref{aux}.\\
\end{remark}

Theorem \ref{teo1} is the one-dimensional version of the group classification carried out by \cite{sv} for the 
poliharmonic equation
\bb\label{1.11}
(-1)^{n}\Delta^{n} u+f(u)=0.
\ee
 In (\ref{1.11}), $\Delta$ is the Laplacian operator, $\Delta^{n}=\Delta^{n-1}\Delta,\,n\in\N$ and $u=u(x),\,x\in\R^k, k \geq 2$. Therefore, one can consider Theorem \ref{teo1} as the extension, to the one-dimensional case, of Svirshcheviskii's results in early 90's.

Later, in \cite{yu1}, the Noether symmetries of (\ref{1.11}) were studied. In the mentioned reference it was shown that the Lie point symmetry group the equation (\ref{1.11}), with $f(u)=u^{p}$, is a Noether symmetry group if and only if
\bb\label{1.12}
p=\f{k+2n}{k-2n}.
\ee

This case, for the poliharmonic equation (\ref{1.11}), corresponds to the largest symmetry Lie algebra for a nonlinear function $f(u)$ with power nonlinearities. Such value of the power is called {\it critical exponent}. For further details, see \cite{yu2}. 

For most cases, the Noether symmetry group is a proper subgroup of the Lie point symmetry group of the corresponding Euler-Lagrange equations. However, there are some examples in the literature where both groups coincide. In \cite{yu2} this fact was firstly discussed. Later, in \cite{yu3}, this point was retaken and many examples were analysed. A considerable number of the examples discussed were related with partial differential equations. However, since that work, more differential equations having the same property have been communicated in the literature.

For instance, in \cite{yi1}, a complete group classification of the semilinear Kohn-Laplace equations was carried out. For that considered family of equations, when the nonlinear term is a power nonlinearity with the critical exponent, all Lie point symmetries are Noether symmetries and consequently, from Noether theorem, it is possible to find conservation laws for them, see \cite{yi2} and \cite{yi3}.

In \cite{yi4} it was shown that for certain semilinear equations on manifolds involving the Laplace-Beltrami operator, the same phenomena occurs, although in that case some restrictions arise from the scalar curvature of the manifold. In \cite{igorjmaa} a concrete application of those results is done and the reader can compare with the classical approach used in \cite{azad}.

In \cite{yi5} the authors studied a family of bidimensional Lane-Emden systems and, for that family, there is a case in which the Lie and Noether symmetry groups coincide. In the case mentioned the power nonlinearities have a relation that the authors called {\it critical line}. Some other examples of the same property can be found in \cite{yi6,yg1,yg2,gilli,renato, moises}.

On one hand, it is well known that the Lie symmetries of equation (\ref{1.2}) are also Noether symmetries and the same property also holds for equation (\ref{1.3}). On the other hand, if one takes, respectively, $n=k=1$ or $k=1$, $n=2$, $f(u)=u^{p}$ and $p$ as in (\ref{1.12}), equation (\ref{1.11}) becomes equation (\ref{1.2}) and (\ref{1.3}), respectively. 

Motivated by these facts, our next result is related with Noether symmetries, which can be stated as the following:
\begin{theorem}\label{teo2}
The Lie point symmetry generator $(\ref{1.6})$ is a Noether symmetry operator of the equation
\bb\label{1.13}
y^{(2n)}+\lambda y^{p}=0,
\ee
with $\lambda\neq 0$, if and only if 
\bb\label{1.14}
p=\f{1+2n}{1-2n}.
\ee
\end{theorem}

We observe that (\ref{1.14}) can easily be obtained from (\ref{1.12}) taking $k=1$, which implies that in this case, equation (\ref{1.1}) inherits similar properties with respect to the Noether symmetries of the poliharmonic equation (\ref{1.11}). In fact, we can formulate another common result, inherited by (\ref{1.1}), which can be announced in the

\begin{theorem}\label{teo3}
All Lie point symmetries of the equation $(\ref{1.13})$, with $p$ given by $(\ref{1.14})$, are Noether symmetries.
\end{theorem}

On one hand, Theorem \ref{teo3} is the one-dimensional version of analogous Bozhkov's results concerning equation (\ref{1.11}), see \cite{yu1}. On the other hand, the same theorem generalizes the results of \cite{go,bok1,bok2} for semilinear ODEs admitting the $\frak{sl}(2,\R)$ symmetry Lie algebra. Therefore, this paper not only extends the results of \cite{sv,yu1} to the equation (\ref{1.1}), but also generalizes the results of \cite{bok1,bok2,ipm} regarding fourth-order equation to an arbitrary even-order semilinear ODEs.

With respect to the other nonlinear cases, except for the operator (\ref{1.4}), all remaining generators are not Noether symmetries.

\section{Lie symmetries and the Noether theorem}\label{rev}

Here we present the necessary elements in order to prove Theorem \ref{teo1}.  Algebraically speaking, a Lie point symmetry of an ordinary differential equation is a one-parameter group of $C^\infty$-automorphisms of the plane preserving the set of integrals of the considered equation.

For a wider and deeper discussion on this subject, the reader is encouraged to check references \cite{bk,ba,i,i1,ol}. These books present the standard construction of the theory. A different viewpoint, more algebraic and focused on ordinary differential equations, can be found in \cite{draisma,oudshoorn}.

In what follows, all functions are assumed to be functions on the universal space ${\cal A}$ of the modern group analysis, see \cite{i1,ib2} and references therein.

\subsection{Lie-B\"acklund operators and symmetries}

An operator
\bb\label{2.1.1}
X=\xi\f{\p}{\p x}+\eta\f{\p}{\p y}+\sum_{i=1}^{\infty}\zeta_{i}\f{\p}{\p y^{(i)}}
\ee
is called {\it Lie-B\"acklund operator} if $\xi=\xi(x,y,y',y'',\cdots),\,\eta=\eta(x,y,y',y'',\cdots)$ and $\zeta_{k}=D(\zeta_{k-1})-y^{(k)}D\xi, \,k\geq 1$, where $\zeta_{0}:=\eta$ and
\bb\label{2.1.2}
D=\f{\p}{\p x}+y'\f{\p}{\p y}+y''\f{\p}{\p y'}+y'''\f{\p}{\p y''}+\cdots
\ee
is the total derivative operator. In this case, if (\ref{2.1.1}) is a Lie-B\"acklund operator, it is usually written in the abbreviated form
\bb\label{2.1.3}X=\xi\f{\p}{\p x}+\eta\f{\p}{\p y}\ee
and the remaining terms are understood.  However, sometimes it is useful to consider the $p-th$ extension of operator  (\ref{2.1.3}), given by 
$$X=\xi\f{\p}{\p x}+\eta\f{\p}{\p y}+\zeta_{1}\f{\p}{\p y'}+\cdots+\zeta_{p}\f{\p}{\p y^{(p)}}.$$

\begin{example}\label{ex1} Operators $(\ref{1.6})$ and $(\ref{1.8})$ are the abbreviated form of the Lie-B\"acklund operators
\bb\label{2.1.4}
D_{p}=x\f{\p}{\p x}+\f{2n}{1-p}y\f{\p}{\p y}+\sum_{k=1}^{\infty}\f{2n+k(p-1)}{1-p}y^{(k)}\f{\p}{\p y^{(k)}}
\ee
and
\bb\label{2.1.5}
X_{3}=x^2\f{\p}{\p x}+(2n-1)xy\f{\p}{\p y}+\sum_{k=1}^{\infty}\left[(2kn-k^2)y^{(k-1)}+(2n-2k-1)xy^{(k)}\right]\f{\p}{\p y^{(k)}}.
\ee
\end{example}

Given an ordinary differential equation $F(x,y,y',\cdots,y^{(n)})=0$, a Lie-B\"acklund operator (\ref{2.1.3}) is called {\it Lie point symmetry generator} if $\xi=\xi(x,y)$, $\eta=\eta(x,y)$ and, for a certain function $\al$ depending on $x,\,y$ and its derivatives, the following identity holds:
\bb\label{2.1.5a}
XF=\al F.
\ee



\begin{remark}
In practical terms, it is not necessary to consider the formal sum $(\ref{2.1.1})$ to obtain the Lie point symmetries of a certain equation. In fact, if the investigated equation is of order $n$, it is enough to consider the $n-th$ extension of the generator and then apply condition $(\ref{2.1.5a})$, which is called {\it invariance condition}. From this constraint, an overdetermined linear system of equations for the coefficients $\xi$ and $\eta$, called {\it determining equations}, will arrise. The solutions of this system give the basis of the Lie point symmetry generators for the considered equation.\\
\end{remark}

Let $X$ be a Lie point symmetry generator of an ordinary differential equation $F=0$. The corresponding Lie point symmetry is a local one-parameter group of transformations $T_{\eps}$ given by 
$$e^{\eps X}(x,y)=\left(x+\sum_{j=1}^{\infty}\f{\eps^j}{j!}X^j x,y+\sum_{j=1}^{\infty}\f{\eps^j}{j!}X^j y\right)=:(\bar{x},\bar{y}),$$
where $\eps$ is taken in a neighbourhood of $0$. 


\subsection{Noether theorem}

The formal sum
\bb\label{2.2.1}
\f{\de}{\de y}=\sum_{j=0}^{\infty}(-1)^j D^j \f{\p}{\p y^{(j)}},
\ee
where $y^{(0)}:=y,\,D^1:=D,\,D^2:=DD,\,D^3:=DDD,\cdots$, is called Euler-Lagrange operator. For each Lie-B\"acklund operator (\ref{2.1.3}) one can associate the Noether operator
\bb\label{2.2.2}
N=\xi+W\f{\de}{\de y'}+\sum_{j=1}^{\infty}D^j(W)\f{\de}{\de y^{(j+1)}},
\ee
where $W:=\eta-y'\xi$.

\begin{example}\label{ex4}
Consider the Lie-B\"acklund operators $(\ref{1.4}),\,(\ref{1.7})$ and $(\ref{1.8})$. Then the Noether operators associated with them are given, respectively, by
\bb\label{n1}
N_{1}=1-\sum_{k=0}^{\infty}y^{(k+1)}\f{\de}{\de y^{(k+1)}},
\ee
\bb\label{n2}
N_{2}=x+\sum_{k=0}^{\infty}\left(\f{2n-2k-1}{2}y^{(k)}-xy^{(k+1)}\right)\f{\de}{\de y^{(k+1)}}
\ee
and
\bb\label{n3}
\ba{lcl}
N_{3}&=&\ds{x^2+\sum_{k=0}^{\infty}\left[k(2n-k)y^{(k-1)}+(2n-2k-1)xy^{(k)}-x^2y^{(k+1)}\right]\f{\de}{\de y^{(k+1)}}}.
\ea
\ee
\end{example}

Ibragimov (see \cite{i1}, Section 8.4, for further details) proved that the Euler-Lagrange operator (\ref{2.2.1}), Lie-B\"acklund (\ref{2.1.3}) and the Noether operators (\ref{2.2.2}) satisfy the {\it Noether identity}
\bb\label{2.2.3}
X+D(\xi)=W\f{\de}{\de y}+DN.
\ee 

An equation $F=0$ has variational formulation if there exists a function ${\cal L}\in {\cal A}$, called {\it Lagrangian}, such that
$$F=\f{\de{\cal L}}{\de y}.$$
In this case, equation
\bb\label{2.2.4}\f{\de{\cal L}}{\de y}=0\ee
is called {\it Euler-Lagrange equation}.

\begin{example}\label{ex4'}
Equation $(\ref{1.1})$ has variational formulation. In fact, consider the Lagrangian
\bb\label{2.2.5}
{\cal L}=\f{(y^{(n)})^2}{2}+F(y),
\ee
where $F'=(-1)^nf$. Then 
$$\f{\de{\cal L}}{\de y}=(-1)^n\left[y^{(2n)}+f(y)\right].$$
\end{example}

\begin{definition}\label{defn}
An operator  (\ref{2.1.3}) is called {\it Noether symmetry operator} of the functional
$$J[u]=\int_{\R}{\cal L}\,dx,$$
if $X{\cal L}+{\cal L}D\xi=DA$, for a certain potential $A\in{\cal A}$. If $A=0$ then the generator $(\ref{2.1.3})$ is called {\it variational symmetry operator}, while for $A\neq const$, it is termed as {\it divergence symmetry operator}. 
\end{definition}

\begin{remark}
Any Noether symmetry is a Lie point symmetry of the corresponding Euler-Lagrange equation,
see \cite{i1,ol}. The potential $A$, also called {\it gauge}, arise from contribution of boundary terms. For further details on Noether symmetries, see \cite{i1,ol}.\\
\end{remark}

Noether's theorem can now be formulated for ordinary differential equations:
\begin{theorem}{\tt (Noether theorem)}\\
For any Noether symmetry $(\ref{2.1.3})$ of the Euler-Lagrange equation $(\ref{2.2.4})$, the quantity
\bb\label{2.2.6}
I=N({\cal L})-A,
\ee
called first integral, is conserved on the solutions of $(\ref{2.2.4})$.
\end{theorem}

\begin{proof} 
Applying the Noether identity to the Lagrangian ${\cal L}$, we have 
$$X{\cal L}+{\cal L}D(\xi)=W\f{\de {\cal L}}{\de y}+DN({\cal L}).$$

On one hand, since $X$ is a Noether symmetry, there exists a function $A\in{\cal A}$, eventually $0$, such that $X{\cal L}+{\cal L}D(\xi)=DA$. On the other hand, on the solutions of the Euler-Lagrange equation (\ref{2.2.4}), one can write
$D(A)=X{\cal L}+{\cal L}D(\xi)=DN({\cal L})$, or $D(N({\cal L})-A)=0.$
Defining $I$ by (\ref{2.2.6}), we can easily see that $DI=0$ on the solutions of (\ref{2.2.5}), that is, $I$ is a first integral, or a conserved quantity, of the Euler-Lagrange equation.
\end{proof}

We would like to recall that constants are trivial first integrals of any equation. Therefore, two first integrals are equivalent if they differ by a constant.

To finish this section, we write (\ref{2.2.6}) explicitly, assuming that ${\cal L}={\cal L}(x,y,y',\cdots,y^{n})$:
\bb\label{2.2.7}
I=\xi{\cal L}+(\eta-y'\xi)\f{\de{\cal L}}{\de y'}+\sum_{j=1}^{n-1}D^{j}(\eta-y'\xi)\f{\de {\cal L}}{\de y^{(j+1)}}-A.
\ee

\section{Auxiliary results}\label{aux}

In this section we prove some technical results that will be useful to prove the main statements of this paper regarding the case $n>1$. Therefore, in the whole section it is presupposed this hypothesis. Although technically a crucial section, the reading of this part can be avoided and we believe that the interested reader could omit it while reading the paper. In fact, the reader can directly go to section \ref{main} and, once one of the results here presented is invoked, one can only consult the requested point. 

However, for those who appreciate technical results or enjoy some manipulation, we begin with
\begin{lemma}\label{lema1}
Let $(\ref{2.1.3})$ be a Lie point symmetry generator of $(\ref{1.1})$. Then its $2n-th$ extension is given by
\bb\label{3.1}
X=\xi(x)\f{\p}{\p x}+[\al(x)y+\be(x)]\f{\p}{\p y}+\sum_{j=1}^{2n}\zeta_{j}\f{\p}{\p y^{(j)}},
\ee
where
\bb\label{3.2}
\zeta_{p}=\be^{(p)}+\al^{(p)}y+\sum_{j=1}^{p}\left[\binom{p}{j}\al^{(p-j)}-\binom{p}{j-1}\xi^{(p-j+1)}\right]y^{(j)}.
\ee
\end{lemma}

\begin{proof}
From \cite{bluman}, it follows that $\xi=\xi(x)$ and $\eta=\al(x)y+\be(x)$, for certain functions $\xi,\,\al$ and $\be$ of $x$. Substituting these expressions into 
\bb\label{3.3}
\zeta_{k}=D(\zeta_{k-1})-y^{(k)}D\xi
\ee
and using induction over $k$, we can easily conclude (\ref{3.2}).
\end{proof}





\begin{lemma}\label{lema4}
The determining equations of $(\ref{1.1})$ are
\bb\label{3.6}
\lambda=\al-2n\xi',
\ee
\bb\label{3.7}
\binom{2n}{k}\al^{(2n-k)}-\binom{2n}{k-1}\xi^{(2n-k+1)}=0,\,\,\,\,1\leq k<2n,
\ee
\bb\label{3.8}
(\al y+\be)f'(y)+\be^{(2n)}+\al^{(2n)}y=\lambda f(y).
\ee
\end{lemma}

\begin{proof}
Let (\ref{2.1.3}) be a Lie point symmetry generator of (\ref{1.1}). By Lemma \ref{lema1}, $X$ takes the form 
$$X=\xi(x)\f{\p}{\p x}+[\al(x)y+\be(x)]\f{\p}{\p y}.$$

From the same Lemma, its $2n-th$ extension is given by (\ref{3.1}), whose remaining coefficients are given by (\ref{3.2}). From the invariance condition (\ref{2.1.5a}) we can write
\bb\label{3.9}
(\al(x)y+\be(x))f'(y)+\zeta_{2n}=\lambda (y^{(2n)}+f(y)).
\ee
From (\ref{3.2}) we can conclude that
\bb\label{3.10}
\zeta_{2n}=\be^{(2n)}+\al^{(2n)}y+\sum_{j=1}^{2n}\left[\binom{2n}{j}\al^{(2n-j)}-\binom{2n}{j-1}\xi^{(2n-j+1)}\right]y^{(j)}.
\ee
Substituting (\ref{3.10}) into (\ref{3.9}), from the coefficient of the terms without derivatives, equation (\ref{3.8}) is obtained. Equation (\ref{3.7}) is obtained from the coefficient of $y^{(2n)}$, while (\ref{3.6}) comes from the remaining coefficients of $y^{(k)},\,\,1\leq k<2n$.
\end{proof}

\begin{lemma}\label{lema5}
Equations $(\ref{3.6})-(\ref{3.8})$ are equivalent to
\bb\label{3.11}
\xi=a_{1}x^2+a_{2}x+a_{3},
\ee 
\bb\label{3.12}
\al=\f{2n-1}{2}(2a_{1}x+a_{2})+k_{1}
\ee
and
\bb\label{3.13}
\left[\f{2n-1}{2}(2a_{1}x+a_{2})+k_{1}\right]yf'(y)+\be(x) f'(y)+\be^{(2n)}(x)+\left[\f{2n+1}{2}(2a_{1}x+a_{2})-k_{1}\right]f(y)=0.
\ee
\end{lemma}

\begin{proof}
From (\ref{3.7}) with $k=2n-1$ and $k=2n-2$, we conclude that $\xi'''=0$, which is equivalent to (\ref{3.11}). Again, from (\ref{3.7}) with $k=2n-1$, we conclude that $\al$ is given by (\ref{3.12}). Then, substituting these expressions to $\al$ and $\xi$ into (\ref{3.6}) and next, into (\ref{3.8}), we arrive at (\ref{3.13}).
\end{proof}

\begin{remark}
We mentioned in Remark \ref{rem6} that the functions listed in Theorem \ref{teo1} arise from a compatibility condition. Such a condition is given by (\ref{3.13}) and it is similar to that one obtained by \cite{sv} in his group classification work regarding equation (\ref{1.11}), as well as in other works related with group classification of semilinear equations, see 
\cite{bok1,yu1,yi1,yi4,ipm,sv}. Therefore, the problem of finding the functions that can enlarge the symmetry groups was already considered in the precedent works. Then in this paper we restrict ourselves to those functions considered in \cite{bok1,yu1,yi1,yi4,ipm,sv}. 
\end{remark}

\section{Group classification}\label{main}
Here we prove Theorem \ref{teo1}. 

\subsection{Case $n=1$}
Let $(\ref{2.1.3})$ be a Lie point symmetry generator of equation
\bb\label{n1.1}
y''+f(y)=0.
\ee
Then, its second extension is given by
$$X^{(2)}=\xi(x,y)\f{\p}{\p x}+\eta(x,y)\f{\p}{\p y}+\zeta_{1}\f{\p}{\p y'}+\zeta_{2}\f{\p}{\p y''},$$
where
$$
\ba{l}
\zeta_{1}=\eta_{x}+y'(\eta_{y}-\xi_{x})-(y')^2\xi_{y},\\
\\
\zeta_{2}=\eta_{xx}+y'(2\eta_{xy}-\xi_{xx})-(y')^2(\eta_{yy}-2\xi_{xy})-(y')^3\xi_{yy}+y''(\eta_{y}-2\xi_{x})-3y'y''\xi_{y}.
\ea
$$

The condition $X^{(2)}(y''+f(y))=\lambda (y''+f(y))$ reads
$$
\eta_{xx}+y'(\eta_{xy}-\xi_{xx})-(y')^2(\eta_{yy}-2\xi_{xy})-(y')^3\xi_{yy}+y''(\eta_{y}-2\xi_{x})-3y'y''\xi_{y}+\eta f'(y)=\lambda (y''+f(y)),
$$
which is equivalent to
\bb\label{n2.1}
\eta_{xx}-f(y)(\eta_{y}-2\xi_{x})+f'(y)\eta+y'[2\eta_{xy}-\xi_{xx}+3f(y)\xi_{y}]+(y')^2(\eta_{yy}-2\xi_{xy})-(y')^3\xi_{yy}=0.
\ee

From the coefficients of $y'$, $(y')^2$ and $(y')^3$, we conclude, respectively, that $2\eta_{xy}-\xi_{xx}+3f(y)\xi_{y}=0$, $\eta_{yy}-2\xi_{xy}=0$ and $\xi_{yy}=0$. Then, solving the last two equations, we obtain 
\bb\label{n3.1}
\xi=a(x)y+b(x),\,\,\eta=a'(x)y^2+c(x)y+d(x).
\ee
Substituting (\ref{n3.1}) into $2\eta_{xy}-\xi_{xx}-3f(y)\xi_{y}=0$ and the remaining part of (\ref{n2.1}), we have
\bb\label{n4}
\ba{l}
2c'(x)-b''(x)+3a''(x)y+3a(x)f(y)=0,\\
\\
d''(x)+(2b'(x)-c(x))f(y)+d(x)f'(y)+c(x)yf'(y)+a'(x)y^2f(y)+c''(x)y+a'''(x)y^2=0.
\ea
\ee

\subsubsection{Case $f(y)=\lambda e^{\al y},\,\lambda\al\neq0$}

Setting $f(y)=\lambda e^{\al y}$ into (\ref{n4}), putting the solution into (\ref{n3.1}) and substituting the solutions into (\ref{2.1.3}) it is obtained a linear combination of the vector fields (\ref{1.4}) and (\ref{1.5}).

\subsubsection{Case $f(y)=\lambda y^{p},\,\lambda\neq0$}

Substituting $f(y)=\lambda y^{p}$ into (\ref{n4}), we arrive at
\bb\label{n5}
\ba{l}
2c'(x)-b''(x)+3a''(x)y+3\lambda a(x)y^p=0,\\
\\
d''(x)+\lambda[2b'(x)+(p-1)c(x)]y^p+\lambda p d(x)y^{p-1}+\lambda a'(x)y^{p+1}+c''(x)y+a'''(x)y^2=0.
\ea
\ee
Now we must separately analyze the cases $p=0,1,-3$. Then, we consider $p$ an arbitrary power if $p\notin\{0,1,-3\}$. 

\begin{enumerate}
\item[a)] $p$ arbitrary

In this case, from (\ref{n5}) we conclude that $a=2c'(x)-b''(x)=2b'(x)+(p-1)c(x)=d=0$. Then, solving the obtained system and substituting the solutions into (\ref{2.1.3}), it is obtained a linear combination of the vector fields (\ref{1.4}) and (\ref{1.6}).

\item[b)] $p=1$

Substituting $p=1$ into (\ref{n4}), we obtain
$$
a''(x)+\lambda a(x)=0,\,\,b''(x)-2c'(x)=0,\,\,a'''(x)+\lambda a'(x)=0,\,\,c''(x)+2\lambda b'(x)=0 
$$
and $d''(x)+\lambda d(x)=0$. Changing $d$ by $\be$, we obtain the generator (\ref{beta0}) with the condition (\ref{beta}). Solving the remaining equations, it is concluded that the solution is a linear combination of the generators (\ref{1.4}), (\ref{1.10}) and (\ref{eqn=1}).

\item[c)] $p=-3$

Substituting $p=-3$ into (\ref{n5}) we conclude that $b=c_{1}x^2+c_{2}x+c_{3}$ and $c=c_{1}x+c_{2}$. Then, solving the remaining equations and putting the solutions into (\ref{2.1.3}), it is then obtained a linear combination of (\ref{1.4}), (\ref{1.7}) and (\ref{1.8}).
\end{enumerate}

\subsubsection{Case $f(y)=\lambda,\,\lambda\in\R$}

In this case, setting $f(y)=\lambda$ into (\ref{n4}), solving the system and substituting the solution into (\ref{2.1.3}) we then obtain a linear combination of the generators (\ref{l1}) -- (\ref{l5}).

\subsection{Case $n>1$}

Now we prove Theorem \ref{teo1} with $n>1$. Actually, it is almost proved in Lemma \ref{lema5}. In what follows we finish the demonstration by considering equation (\ref{3.13}) and its consequences on (\ref{3.11}) and (\ref{3.12}).

\subsubsection{Case $f(y)=\lambda e^{\al y},\,\lambda\al\neq0$}

Substituting $f(y)=\lambda e^{\al y}$ into (\ref{3.13}) we conclude that $a_{1}=0, \beta = 0$ and $k_{1}=(1-2n)a_{2}/2$. Then we have $\xi=a_{2}x+a_{3}$ and $\eta=-2na_{2}\al$. Substituting these coefficients into (\ref{2.1.3}) it is obtained a linear combination of the vector fields (\ref{1.4}) and (\ref{1.5}).

\subsubsection{Case $f(y)=\lambda y^{p},\,\lambda\neq0$}

Substituting $f(y)=\lambda y^{p}$ into (\ref{3.13}), we arrive at
\bb\label{3.14}
\left\{[(2n-1)p+(2n+1)]a_{1}x+\f{(2n-1)p+(2n+1)}{2}a_{2}+(p-1)k_{1}\right\}y^p+p\be y^{p-1}+\be^{(2n)}=0.
\ee
We shall now consider the cases $p$ arbitrary, $p=1$ and $(\ref{1.14})$. When $p=0$, we have $f(y)=\lambda$.

\begin{enumerate}
\item[a)] $p$ arbitrary

From (\ref{3.14}), if $p$ is arbitrary, then $\be=a_{1}=0$ and
$$k_{1}=\f{(2n-1)p+(2n+1)}{2(1-p)}a_{2}.$$
Therefore
$$\xi=a_{2}x+a_{3},\,\,\,\eta=\f{2n}{1-p}a_{2}y$$
and, once these components are substituted in (\ref{2.1.3}), one obtains a linear combination of (\ref{1.4}) and  (\ref{1.6}).

\item[b)] $p=1$

For $p=1$, (\ref{3.14}) implies that $a_{1}=a_{2}=0$ and $\be^{(2n)}+\lambda\be=0,$ which gives us the generators (\ref{1.4}), (\ref{1.10}) and (\ref{beta0}).

\item[c)] $p=\f{1+2n}{1-2n}$

Finally, setting $p=\f{1+2n}{1-2n}$ into (\ref{3.14}), one concludes that $k_{1}=\be=0$ and then
$$\xi=a_{1}x^2+a_{2}x+a_{3},\,\,\,\eta=(2n-1)a_{1}xy+\f{2n-1}{2}a_{2}y$$
and, once substituted into (\ref{2.1.3}), it is obtained a linear combination of the generators (\ref{1.4}), (\ref{1.7}) and (\ref{1.8}).
\end{enumerate}

\subsubsection{Case $f(y)=\lambda,\,\lambda\in\R$}

Substituting $f(y)=\lambda$ in (\ref{3.13}), we conclude that
\bb\label{4.2.4.1}
\be(x)=\sum_{k=0}^{2n-1}c_{k}x^k-a_{1}\lambda\f{x^{2n+1}}{(2n)!}-\f{2n+1}{2}a_{2}\lambda\f{x^{2n}}{(2n)!}+k_{1}\lambda\f{x^{2n}}{(2n)!}.
\ee

Therefore, substituting (\ref{4.2.4.1}), (\ref{3.11}) and (\ref{3.12}) into (\ref{2.1.3}), it is obtained a linear combination of operators (\ref{1.4}), (\ref{l1}), (\ref{l2}), (\ref{l3}) and (\ref{l4}).

\section{Proofs of Theorems \ref{teo2} and \ref{teo3}}\label{noether}

Here we classify, from the Lie point symmetries, those which are Noether symmetries, according to Definition \ref{defn}. We restrict ourselves, however, only to the nonlinear cases.

To begin with, it is very simple to conclude that for the generator (\ref{1.4}), the following identity holds
\bb\label{4.1}
X_{1}{\cal L}+{\cal L}D\xi=0,
\ee
where ${\cal L}$ is the Lagrangian (\ref{2.2.5}), for any smooth function $F=F(y)$. Therefore the translation in $x$ is a Noether symmetry operator to equation (\ref{1.1}).

\subsection{Proof of Theorem \ref{teo2}}
Let us now prove Theorem \ref{teo2}. Firstly, applying (\ref{2.1.4}) to the Lagrangian (\ref{2.2.5}), with
$$F(y)=(-1)^{n}\f{\lambda}{p+1} y^{p+1},$$
we have
$$D_{p}{\cal L}+{\cal L}D\xi=\f{2n+1+p(2n-1)}{1-p}{\cal L}.$$
Therefore, $D_{p}$ is a variational symmetry operator if and only if $p$ is given by (\ref{1.14}). This proves Theorem \ref{teo2}.

\subsection{Proof of Theorem \ref{teo3}}
We have already demonstrated that the generators $X_{1}$ and $X_{2}$, given respectively by (\ref{1.4}) and (\ref{1.7}), are variational symmetries operators. Then, in order to prove Theorem \ref{teo3} it is only necessary to prove that $X_{3}$, given by (\ref{1.8}), is a Noether symmetry operator.

Applying the operator (\ref{2.1.5}) to the Lagrangian
\bb\label{4.2}
{\cal L}=\f{(y^{(n)})^{2}}{2}+(-1)^{n}\lambda \f{1-2n}{2} y^{\f{2}{1-2n}}
\ee
we obtain
\bb\label{4.3}
X_{3}{\cal L}+{\cal L}D\xi=n^2y^{(n-1)}y^{n}=D\left[\f{n^2}{2}(y^{(n-1)})^2\right].
\ee
Therefore $X_{3}$ is a divergence symmetry operator with potential
\bb\label{4.4}
A=\f{n^2}{2}(y^{(n-1)})^2.
\ee
\section{First integrals and exact solutions}\label{integrals}

From physical point of view, a first integral corresponds to a constant of motion, that is, a quantity which is preserved along time. Now we establish first integrals associated with the Noether symmetries of the nonlinear cases. An essential point to understand the results presented here is that
\bb\label{5.1}
\f{\de {\cal L}}{\de y^{(k)}}=(-1)^{n-k}y^{(2n-k)}, k \geq 1
\ee
where ${\cal L}$ is the Lagrangian (\ref{2.2.5}).

Considering the Lie point symmetry generator $X_{1}$, given by (\ref{1.4}), a first integral can be found setting on (\ref{2.2.6}) the Noether operator (\ref{n1}) associated with the generator (\ref{1.4}). Then, considering (\ref{5.1}), a first integral, for any smooth function $F=F(y)$ in (\ref{2.2.5}) is given by
\bb\label{5.2}
I=\f{(y^{(n)})^2}{2}+F(y)+\sum_{j=0}^{n-1}(-1)^{n-j}y^{(j+1)}y^{(2n-j-1)}.
\ee
In particular, considering the Lagrangian (\ref{4.2}), we have
\bb\label{5.3}
I_{1}=\f{(y^{(n)})^2}{2}+(-1)^{n}\lambda \f{1-2n}{2} y^{\f{2}{1-2n}}+\sum_{j=0}^{n-1}(-1)^{n-j}y^{(j+1)}y^{(2n-j-1)}.
\ee

Let us now find the first integral associated with the variational symmetry operator $X_{2}$, given by (\ref{1.5}). Replacing the operator $N$ in equation (\ref{2.2.6}) by the Noether operator given by (\ref{n2}), a simple calculation yields
\bb\label{5.4}
\ba{lcl}
I_{2}&=&\ds{x\f{(y^{(n)})^2}{2}+(-1)^n\lambda x\f{(1-2n)}{2}y^{\f{2}{1-2n}}}\\
\\
&&\ds{+\sum_{j=0}^{n-1}(-1)^{n-j-1}\left(\f{2n-2j-1}{2}y^{(j)}-xy^{(j+1)}\right)y^{(2n-j-1)}.}
\ea
\ee

Finally, considering (\ref{n3}), the Lagrangian (\ref{4.2}) and the potential (\ref{4.4}), we obtain the third first integral
\bb\label{5.5}
\ba{lcl}
I_{3}&=&\ds{\f{x^2}{2}(y^{(n)})^2+(-1)^n\lambda x^2\f{1-2n}{2}y^{\f{2}{1-2n}}-\f{n^2}{2}(y^{(n-1)})^2}\\
\\
&&\ds{+\sum_{j=0}^{n-1}(-1)^ {n-j-1}\left[j(2n-j)y^{(j-1)}+(2n-2j-1)xy^{(j)}-x^2y^{(j+1)}\right]y^{(2n-j-1)}}.
\ea
\ee


 
From  (\ref{5.3}), (\ref{5.4}) and (\ref{5.5}), once considering the expression $x^2 I_{1}-2x I_{2}+I_{3}$ and after reckoning, it is obtained the following ordinary differential equation
\bb\label{7.1}
\sum_{j=0}^{n-1}(-1)^{n-j}\left[j(2n-j)y^{(j-1)}y^{(2n-j-1)}\right]+\f{n^2}{2}(y^{(n-1)})^2+x^2 I_{1}-2x I_{2}+I_{3}=0.
\ee

For the case $n=1$, we have
\bb\label{7.2}
x^2 I_{1}-2x I_{2}+I_{3}=-\f{y^2}{2},
\ee
which can easily provide a solution of the equation (\ref{1.2}). Such solution, as far as we know, was firstly found by \cite{pin}. For further details regarding the Ermakov equation, see \cite{leandri}.

For $n=2$, (\ref{7.1}) is reduced to
$$-3yy''+2(y')^2+x^2 I_{1}-2x I_{2}+I_{3}=0,$$
an equation first obtained in \cite{bok1} and later discussed in \cite{bok2}.

We observe that from (\ref{7.1}) it is possible to obtain a three-parameter family of solutions of the considered equation, which does not imply that it is an easy task. For instance, to the case $n=1$ a solution is implicitly given by (\ref{7.2}). However for $n=2$ the situation is a little bit more complicated, see, for instance, the discussions about this point in \cite{bok1,bok2}.


As it was previously pointed out, it is interesting to observe that all Lie symmetries of the considered equation are also Noether symmetries. This fact was observed in the literature for certain equations, see \cite{yu2,yu3}. However, for equations of the type (\ref{1.1}), this fact was known for the case $n=1$ and, more recently, for the case $n=2$. In \cite{bok1,bok2} these aspects were discussed, but not from the point of view of the present paper. An implicit solution of (\ref{1.3}) was presented in \cite{bok1}, see also \cite{bok2}. In \cite{ipm} it was presented an explicit three-parameter family of solutions of (\ref{1.3}) by using a general linear combination of the Lie point symmetry generators of the equation.

We shall now use the characteristic method (see \cite{bk}, page 169, and \cite{i1}, section 9.4, for further details) to obtain a three-parameter family of solutions of
\bb\label{7.3}
y^{(2n)}+\lambda y^{\f{1+2n}{1-2n}}=0.
\ee

Firstly, we recall that the Lie point symmetry generators admitted by (\ref{7.3}) are linear combinations of the generators (\ref{1.4}), (\ref{1.7}) and (\ref{1.8}). Secondly, let $X=\al X_{1}+2\be X_{2}+\gamma X_{3}$, where $\al,\,\be$ and $\gamma$ are constants, be a linear combination of them. We therefore obtain
$$X=(\al +2\be x+\gamma x^2)\f{\p}{\p x}+(2n-1)(\be y+\gamma xy)\f{\p}{\p y}.$$

Solving the corresponding characteristic equations
$$\f{dx}{\al +2\be x+\gamma x^2}=\f{dy}{(2n-1)(\be y+\gamma xy)}$$
we have the invariant 
$$\phi=\f{y}{(\al+2\be x+\gamma x^2)^{\f{2n-1}{2}}}.$$

Then, taking 
$$y=A_{n}(\al+2\be x+\gamma x^2)^{\f{2n-1}{2}}$$
and imposing that this function is a solution of (\ref{7.3}) we conclude that
$$A_{n}=\left[(-1)^{n+1}\f{\lambda}{(\be^2-\al\gamma)^{n}}\left(\f{2^{n}n!}{(2n)!}\right)^{2}\right]^{\f{2n-1}{4n}}.$$

Then
\bb\label{7.4}
y(x)=\left[(-1)^{n+1}\f{\lambda}{(\be^2-\al\gamma)^{n}}\left(\f{2^{n}n!}{(2n)!}\right)^{2}\right]^{\f{2n-1}{4n}}(\al+2\be x+\gamma x^2)^{\f{2n-1}{2}}
\ee
is a three-parameter family of solutions of (\ref{7.3}) since $\be^2-\al\gamma\neq0$ and 
$$(-1)^{n+1}\f{\lambda}{(\be^2-\al\gamma)^{n}}>0.$$ 
In particular,
$$
y(x)=\left(\f{\lambda}{\be^2-\al\gamma}\right)^{\f{1}{4}}\sqrt{\al+2\be x+\gamma x^2}
$$
is a solution of the Ermakov equation (\ref{1.2}), while
$$
y(x)=\left[\f{-\lambda}{9(\be^2-\al\gamma)^2}\right]^{\f{3}{8}}(\al+2\be x+\gamma x^2)^{\f{3}{2}}
$$
is a solution of (\ref{1.3}), which was obtained, changing constants, in \cite{ipm}.

\section{Conclusions}

In the present paper we extended some results presented in \cite{bok1,ipm} for an arbitrary even-order autonomous ordinary differential equation and we also presented the one-dimensional version of the results obtained in \cite{sv,yu1}. 

Additionally, the class of equations (\ref{1.1}) was investigated from the point of view of the modern group analysis. It was shown that the largest symmetry Lie algebra, for the nonlinear cases, is reached for equation (\ref{7.3}). Moreover, for this equation the Noether symmetry group coincides with the Lie point symmetry group, an uncommon case in the literature. This fact was well known for equations (\ref{1.2}) and (\ref{1.3}), see \cite{bok1,bok2,moises}. However, for the class (\ref{1.1}), for arbitrary $n$, this is the first paper communicating such a mentioned result. We can explain this mysterious and interesting fact as a phenomena that occurs for certain equations involving power nonlinearities, and for a specific exponent the number of Lie symmetries not only is the largest, compared with other nonlinear cases of the group classification, but also all Lie symmetries are Noether symmetries. This fact, up to our knowledge, was first emphasised in \cite{yu2}. In analysis, such kind of exponent, called {\it critical exponent}, is related not only with embedding theorems, but also with some values dividing existence and non-existence cases of solution for certain differential equations. The reader is guided to \cite{yu2} for a better discussion. The question now is: why does this phenomena, that is, the Lie symmetry group coincides with the Noether symmetry group, only occur for this kind of nonlinearity for the class investigated? One possible explication might come from the the existence of nontrivial conserved quantities obtained from all Lie point symmetries of these equations, since for this kind of equations we not only have the maximal symmetry Lie algebra, but from any symmetry one can establish a conserved quantity. In particular, we do not know in the literature a semilinear Euler-Lagrange equation of the type (\ref{1.1}) with the following properties:
\begin{enumerate}
\item For the nonlinear cases the symmetry Lie algebra is the largest;
\item All Lie point symmetries are Noether symmetries;
\item At least one symmetry provides a trivial conserved quantity.
\end{enumerate} 

Then last, but not least, we leave a question: are there equations or systems accomplishing the points above?

\section*{Acknowledgment}

The authors would like to thank the reviewers for their useful and kind comments, which improved the paper. We also show our gratitude to Dr. Stylianos Dimas, for his careful reading of the manuscript, as well as useful discussions and suggestions.

The authors would like to thank FAPESP, scholarship n. 2012/22725-4 and grant n. 2014/05024-8, and UFABC for financial support. I. L. Freire is also grateful to CNPq for financial support, grant n. 308941/2013-6. P. L. da Silva also would like to thank CAPES for the PhD scholarship provided.

\end{document}